\documentclass[11pt]{amsart}
\usepackage[latin1]{inputenc}
\usepackage[T1]{fontenc}
\usepackage{latexsym}
\usepackage{amsmath,amsthm}
\usepackage{amsfonts}
\usepackage{epsfig}
\usepackage{amssymb}
\usepackage{graphicx}
\usepackage{amscd}
\usepackage[all]{xy}
\usepackage{enumerate}

\input xy
\xyoption{all}
\newtheorem{lemma}{Lemma}[section]
\newtheorem{theorem}[lemma]{Theorem}
\newtheorem{corollary}[lemma]{Corollary}
\newtheorem{proposition}[lemma]{Proposition}

\theoremstyle{definition}
\newtheorem{definition}{Definition}[section]
\newtheorem{example}{Example}[section]
\numberwithin{equation}{section}
\theoremstyle{remark}
\newtheorem{Rem}{Remark}[section]

\newcommand\R{\mathbb R}
\newcommand\C{\mathbb C}
\newcommand\Hl{\mathbb H}
\newcommand\hp{\Hl\text{P}}
\newcommand\cp{\C\text{P}}
\newcommand\rp{\R\text{P}}
\newcommand\delbar{\overline{\partial}}

\newcommand\So{\operatorname{SO}}
\newcommand\Su{\operatorname{SU}}

\newcommand\Sl{\operatorname{SL}_2(\C)}
\newcommand\OO{\operatorname{O}}

\newcommand\g{{\frak g}}

\newcommand\su{{\frak su}}

\newcommand\la{\lambda}

\newcommand\aaa{\alpha}
\renewcommand\b{\beta}
\newcommand\x{\bold x}
\newcommand\y{\bold y}
\newcommand\p{\bold p}
\newcommand\q{\bold q}
\newcommand\X{\bold X}

\title{Isomonodromic Deformations and $\Su_2-$Invariant Instantons on $S^4$}
\author{Richard Mu\~{n}iz Manasliski}

\address{{Centro de Matemática, Facultad de Ciencias},
Iguá 4225 esq. Mataojo C.P. 11400, Montevideo, Uruguay}
\email{rmuniz@cmat.edu.uy}

\begin{document}

\begin{abstract} Anti-self-dual (ASD)  solutions to the Yang-Mills equation (or
instantons) over an anti-self-dual four manifold, which are invariant
under an appropriate  action of a three
dimensional Lie group, give rise, via twistor construction, to isomonodromic 
deformations of
connections on $\cp^1$ having four simple singularities. As is well known
this kind of deformations is governed by the sixth Painlev\'e equation P{\sc
vi}$(\alpha,\beta,\gamma,\delta)$.
  We work out the particular case of the 
$\operatorname{SU}_2$-action on $S^4$, obtained from the irreducible 
representation
on $\R^5$. In particular, we express the parameters
$(\alpha,\beta,\gamma,\delta)$ in terms of the instanton number. The present 
paper contains the proof of the result anounced in ~\cite{Mu}. 
\end{abstract}

\maketitle

\section{Introduction}

 Instantons are particular solutions
 to the Yang-Mills equation, in dimension four, and if they
 admit a sufficiently
 large symmetry group (a three dimensional one acting with
 cohomogeneity one) it is possible to make  a
 one dimensional reduction of the self-duality equation.
We will focus on the fact that these instantons, for an appropriate
 action of a Lie group, are related to an ordinary differential
 equation on the complex domain, which is known as the sixth
 Painlev\'e equation. We denote it by P{\sc vi}$(\alpha,\beta,\gamma,\delta)$ 
and it is given by
\begin{align*} 
      \frac{d^2y}{dx^2}=
                               &\frac12\left(\frac1y+\frac1{y-1}+\frac1{y-x}\right)
                               \left(\frac{dy}{dx}\right)^2-
                               \left(\frac1x+\frac1{x-1}+\frac1{y-x}\right)
                               \left(\frac{dy}{dx}\right)+\\   
                         &\left(\alpha+\beta\frac{x}{y^2}+
                      \gamma\frac{x-1}{(y-1)^2}+\delta\frac{x(x-1)}{(y-x)^2}\right)\frac{y(y-1)(y-x)}{x^2(x-1)^2},   
\end{align*}
where $\aaa$, $\beta$, $\gamma$ and $\delta$ are complex parameters.
This equation was found by Richard Fuchs in 1907 and is very
important in the ODE's theory over the complex domain. Its general
solutions are transcendental functions and its key
property is that it does not have {\it movable critical points}, what is
usually called {\it Painlev\'e's property} ~\cite{Con} . The method used by Fuchs,
which is the way P{\sc vi} appears here, is that of {\it
isomonodromic deformations}.

Our research is inspired by the work of Nigel Hitchin
~\cite{Hit1}, where P{\sc vi} is related to some new examples of Einstein 
metrics. The basic idea is to use the symmetry 
group to produce a
connection on the twistor space with {\it logarithmic singularity}
along an anticanonical divisor. Since the twistor space is a fibre
bundle with fibre $\cp^1$ over the real 4-manifold, and the divisor
will intersect a fibre generically in four distinct points, the connection gives
rise to an isomonodromic deformation of connections on $\cp^1$ having
four simple poles. This kind of deformations is governed by the sixth
Painlev\'e equation, which was proved in ~\cite{JiMi}, and this gives us the
relation with invariant instantons mentioned before.

 We apply the construction to the example of 
$\Su_2$ acting on $S^4$ via the irreducible representation on
$\R^5$; this action has two exceptional orbits of dimension two and
all the others are three dimensional. The $\Su_2$-invariant instantons
are classified by the integer numbers of the form $n\equiv 1\text{mod }4$. 
For each $n$ the
Chern number of the instanton is $c_2=\frac{n^2-1}{8}$ (see ~\cite{Gil1}
or ~\cite{Gil2}). We will prove that the parameters of P{\sc vi} related
to these instantons are given by
\begin{align*}
\aaa^{\pm} &=\tfrac{1}{8} (n\pm 2)^2 &\beta &=-\tfrac{1}{8} n^2\\
\gamma &=\tfrac 18n^2        &\delta   &=-\tfrac18(n^2-4).
\end{align*}

It is known that for these parameters the equations are algebraically
related (~\cite{Oka}). These values for the parameters were conjectured
by Jan Segert in ~\cite{Seg}, but there has been no proof in the literature
until now. Our point of view here is quite different from that of
~\cite{Seg}, since we use a more geometric approach. 

 We also make some explicit calculations by means of which it is possible
 to express the solution to Painlev\'e's equation in terms of the
 corresponding instanton.

Finally, we would like to mention here that it is possible to look at
instantons which are singular along the exceptional orbits. For
instance instantons having a holonomic singularity as those considered
by Kronheimer and Mrowka in ~\cite{KronMro}.  It is not difficult to show
what the corresponding parameters would be for these instantons in
terms of the holonomy parameter. Unfortunatly we are yet unable to prove
 that such instantons exist with the property of being invariant
under the action, even in the case of $S^4$.

\section{Preliminaries}

In this section we remember the definition of the singular connection 
 on the Penrose-Ward
 transform of an anti-self-dual (ASD) vector bundle, induced by the
 action of a three
 dimensional Lie group $G$. We review some of its properties, in
 particular that the singularity is logarithmic along an anticanonical
 divisor of the twistor space. This connection gives rise to an
 isomonodromic deformation of connections on $\cp^1$, and this leads us
 to the sixth Painlev\'e equation. 
 
\subsection{Cohomogeneity one ASD manifolds and their
 associated connection}

In what follows $M$ will denote an ASD (in
particular four dimensional) manifold, so we have its
twistor space $Z$, and $G$ will be a compact three dimensional Lie
group (we can think for example in $\Su_2$). We suppose that $G$ acts
on $M$ having three dimensional principal orbit and preserving the
conformal structure.  In particular, the principal stabilizer
$\Gamma\subset G$ is discrete and $M/G$ is diffeomorphic to $\R$,
$S^1$, $[0,1)$ or $[0,1]$. We can always take a section $M/G\rightarrow
M$ which intersects all the orbits perpendicularly. Such an action
induces an action  on the twistor
space by biholomorphisms, that preserves the real structure and the
twistor lines. At the infinitesimal level we can write the action as a map
 $\alpha\colon Z\times{\frak g}\rightarrow TZ$, which when complexified is
 a homomorphism
\begin{equation}\alpha\colon Z\times{\frak g}_{\C}\rightarrow TZ.
\end{equation}
(We denote by  $TZ$ the holomorphic tangent bundle of $Z$.)
 For the following definitions to make sense we will assume that $\aaa$
 has rank three.

Set $Y=\{z\in Z: rk(\aaa_z)\leq 2\}\subset Z$. By looking at $\aaa$ as a
section of $TZ\otimes{\frak g}_{\C}$ we have that
 $s:=\Lambda^3\aaa\in H^0(Z,\Lambda^3TZ\otimes\Lambda^3{\frak
 g}_{\C})\cong H^0(Z,K^{-1}_Z)$, is a section of the anticanonical
 bundle of $Z$. The set $Y$ is then exactly the zero locus of $s$, and
 it is nonempty because the anticanonical divisor has degree four
 restricted to each real twistor line. Then $Y$ is a codimension
 one analytic subvariety of $Z$; we denote by $Y^{\circ}$ the smooth
 part of $Y$.

 From the fact that the action preserves the real structure we
 conclude that $\aaa$ is compatible with this structure, in the sense
 that the following identity is satisfied 
$$d\tau\circ\aaa_z=\aaa_{\tau(z)}.$$
This implies that $s$ restricted to
a line
vanishes identically, has two double zeroes in two antipodal points, or
vanishes non degenerately in four points forming antipodal pairs. For a more 
detailed explanation see ~\cite{Hit1}.

Let us now consider a Hermitian vector bundle $E$ over $M$ endowed with
an ASD connection $\mathcal D$, and structure group $S=\Su_2$, together 
with a lift
of the $G$-action by homomorphisms leaving the connection invariant. Denote 
by $\tilde{E}$ the holomorphic vector bundle over $Z$ induced by the pair 
$(E, \mathcal D)$; usually known as Penrose-Ward transform ~\cite{AHS}.
It is possible to define a connection $\nabla$ in
 $\tilde{E}\big|_{Z\setminus Y}$, in the following way (as is done 
in ~\cite{MaWoo})
:
$$
\nabla_{\aaa(\mathbf X)}:={\mathcal L}_{\mathbf X}\qquad\forall\,\mathbf X\in\frak{g}_{\C}.
$$
Where ${\mathcal L}_{\mathbf X}$ is Lie derivative of sections.
It is not difficult to verify the integrability of this conection.

\begin{proposition} The connection $\nabla$ defined above in
$\tilde{E}\big|_{Z\setminus Y}$ is flat.
\end{proposition}

\begin{proof}  The statement has a simple geometric interpretation. The Lie
      algebra $\mathfrak g_{\C}$ acts on $\tilde E$, which means that we
           have a Lie algebra homomorphism $\tilde{\aaa}\colon \mathfrak
           g_{\C}\rightarrow\mathfrak X(\tilde E)$. The horizontal
           distribution of the connection $\nabla$ is simply the image
           of $\tilde{\aaa}$, from which we conclude its integrability
           since the Frobenius condition is necessarily satisfied. When the
           $G$ action can be extended to a $G^{\C}$ action the
           horizontal foliation is given by its orbits and the
           connection is $G^{\C}$-invariant.\end{proof}

By its definition $\nabla$ satisfies a compatibility condition with
the real structure. These means that if
$\sigma\colon\tau^{\ast}\bar{\tilde E}\rightarrow\tilde{E}^{\ast}$, is the 
isomorphism 
given by ~\cite{AHS} theorem ?, which can be expessed as
$\sigma=h\circ\hat\tau$, where $\hat\tau$ is the obvious lifting of
$\tau$ and $h\colon\bar{\tilde E}\rightarrow \tilde E^{\ast}$ is the
isomorphism induced by the Hermitian structure, then  
\begin{equation}\sigma(\tau^{\ast}\overline{\nabla})=\nabla^{\ast}.
\end{equation}
Or in coordinates,
\begin{equation}\tau^{\ast}A=-A^{\ast};\end{equation} 
where $A$ is the connection 1-form.

\begin{example} Let $M$ be a $G$-invariant conformal ASD manifold, such that
      $\aaa$ has rank three. We can think about the examples 2.1 and 2.2,
      or in ~\cite{Hit1}  more examples for $G=\Su_2$ can be found.

   Let  $(V,\rho)$ be a two dimensional representation of $G$, and let us
   take the trivial vector bundle over $M$, $E=M\times V$,
   with the trivial flat connection. The $G$ action on $E$ is simply
   the diagonal one
$$ g\cdot(x,v)=(gx,\rho(g)v).
$$   
Then $\tilde{E}=Z\times V$, and the connection 1-form of $\nabla$
is given by 
$$-\dot\rho\circ\aaa^{-1}\colon T(Z\setminus Y)\rightarrow\text{End}(V).$$ 
This can be seen by using the definition; more precisely, if $g(t)$ is
a curve in $G$ such that  $g(0)=I$, $\dot g(0)=\X$, then it follows that 
$$
   \nabla_\X s(z)={\mathcal L}_{\aaa^{-1}(\X)}s(z)=\frac{d}{dt}\Big|_{t=0}\rho(g(t)^{-1})s(g(t)z)=(\X-\dot\rho(\aaa^{-1}(\X)))s(z).
$$
The above is the connection considered by Hitchin in ~\cite{Hit1}, for
$G=\Su_2$ and $(V,\rho)$ the canonical representation.
\end{example}

\subsection{Logarithmic singularities}

What is important about the connection  $\nabla$ is its behaviour near the
singularity. It has what is called logarithmic singularity along
$Y^{\circ}$, which means that $\nabla$ has a simple pole along
$Y^{\circ}$ and that it induces a flat connection on
$\tilde{E}\big|_{Y^{\circ}}$ without singularities. We essentially
follow Malgrange's point of view in ~\cite{Mal} about this subject.

\begin{definition}
(1) A differential form $A$ in a complex manifold $X$ is
                 said to have a {\it simple pole\/} along a smooth
                 hypersurface $Y$, if it can be extended to a holomorphic
      section of $\Lambda^kT^{\ast}X\otimes[Y]$; here $[Y]$ is the line
      bundle asociated to $Y$.

(2) Let $i\colon Y\hookrightarrow X$ be the inclusion
         map. We say that a 1-form $A$ has {\it logarithmic
         singularity\/} along $Y$ if it has a simple pole
         there and $i^{\ast}A=0$. We are viewing $A$ as an
         element in $H^0(X,T^{\ast}X\otimes[Y])$ with the
         obvious extension of the map $i^{\ast}$.
\end{definition}

The above definition can be expressed locally as follows. Let $A$ be a
1-form with a simple pole along $Y$, and let us take a coordinate system
$(z_1,\ldots, z_n)$ in such a way that $Y$ is determined by
$\{z_1=0\}$. In these coordinates
$A=\sum_ja_jdz_j$
, and as it has simple pole on $Y$  $a_j=\frac{b_j}{z_1}$, with
$b_j$ holomorphic for all $j$. Therefore
$i^{\ast}A=\sum_jb_j|_Yi^{\ast}dz_j=\sum_{j\geq
2}b_j|_Yd(z_j|_Y)$, from which $a_j$ is holomorphic for $j\geq 2$. 
Namely,
$$
A=\frac{b_1}{z_1}dz_1+\sum_{j\geq 2}a_jdz_j,$$
with $\sum_{j\geq 2}a_jdz_j$ holomorphic. This is equivalent to
saying that $A=\frac{B}{z_1}$,  
$B$ being holomorphic and $B|_{TY}=0$. It can also be easily seen that $A$
has logarithmic singularity along $Y$ if and only if $A$ and $dA$ both
have simple poles there.

Let us supose now that $E$ is a vector bundle over $X$ endowed with an
integrable connection defined on $X\setminus Y$ (with $X$ and $Y$ as
in definition 3.1). Let $A$ be the
connection 1-form associated to a local frame of $E$.
\begin{definition} We say that $\nabla$ has a {\it simple pole} or a {\it logarithmic
 pole} along  $Y$ if $A$ has a simple pole or a logarithmic pole along $Y$
 respectively
\end{definition}

\begin{Rem} If $\nabla$ is a connection in $E\rightarrow X$ with
           a logarithmic singularity along $Y$, and $f\colon
           Z\rightarrow X$ is a holomorphic function which is
           transversal to $Y$, then $f^{\ast}\nabla$ is a connection
           with logarithmic singularity along $f^{-1}(Y)$.
\end{Rem}

\begin{theorem}[Hitchin] If $s$ vanishes non degenerately, then the
          1-form $\aaa^{-1}\colon TZ\rightarrow\mathfrak g_{\C}$ has
	  logarithmic singularity along the smooth part of $Y$.
\end{theorem}

From the above  proposition together with proposition 2.2 one can
easily deduce the following result:

\begin{corollary} If $s$ vanishes non degenerately, the connection associated
to the action of $G$ has a logarithmic singularity along the smooth
part of $Y$.
\end{corollary}
 
 From now on we will consider only actions such that $s$ vanishes non
 degenerately, since we are interested in connections having
 a logarithmic singularity.

In a generic twistor line $P$ the connection 1-form can be written as 
\begin{equation}
A_P(z)dz=\left(\frac{A_1(P)}{z-z_1}+\frac{A_2(P)}{z-z_2}+\frac{A_3(P)}{z-z_3}+
\frac{A_4(P)}{z-z_4}\right)dz,\end{equation}
where the $A_i(P)\text{'s}\in\g_{\C}$ satisfy $\sum_i
A_i=0$, and $(z_1,z_2)$ $(z_3,z_4)$ are antipodal pairs. 
(We are implicitly using an identification of $P$ with $\cp^1$.)
In view of the compatibility with the real structure, the following 
relations between the residues hold
\begin{equation} A_2=-A^{\ast}_1,\qquad A_4=-A_3^{\ast}.\end{equation}

\subsection{Isomonodromic deformations and the
sixth Painlev\'e  equation}

By a monodromy representation on a space $B$ we mean a representation of
the fundamental group of $B$, with values in some group $\mathcal G$,
which is compatible with change of base point isomorphisms. An
isomonodromic deformation is just a family of monodromy
representations which is constant modulo conjugation. More precisely,
by an isomonodromic deformation on $F$ we mean a fiber bundle with
fiber $F$, base space $X$ and total space $B$, together with a
monodromy representation $(\rho_x,{\mathcal G})$ for each $x\in X$ such that
the translation homomorphisms ~\cite{Stee} are given by conjugation on
$\mathcal G$. 

The family of connections $\nabla\big|_P$ parametrized by twistor lines 
defines an isomonodromic deformation (see ~\cite{Woo}). As is shown in 
~\cite{Hit1} 
the four points of intersection of a generic line with $Y$ have nonconstant 
cross ratio (as we change the line). Then, we can assume the intersection 
points to be $\{0, 1, x, \infty\}$ being $x$ the cross ratio, after an 
appropiate M\"oebius transformation. Doing so, the 
connection is given by  

$$A(x,\zeta)d\zeta=\left(\frac{A_1(x)}{\zeta}+\frac{A_2(x)}{\zeta-1}+
\frac{A_3(x)}{\zeta-x}\right)d\zeta.$$

By the isomonodromic property the residues have to satisfy the Schlesinger's 
equations

\begin{align*}
\frac{dA_1}{dx}&=\frac{[A_1,A_3]}{x}\\
\frac{dA_2}{dx}&=\frac{[A_2,A_3]}{x-1}\\
\frac{dA_3}{dx}&=-\frac{[A_1,A_3]}{x}-\frac{[A_1,A_3]}{x-1};\end{align*}
the last one simply says that
$A_1+A_2+A_3=-A_{\infty}=constant$. Note that the functions
$\text{tr}(A_i^2)$ are constant in the deformation. We refer to ~\cite{Mal} 
for the details, see also ~\cite{Hit1, Hit2, Hit3}.

The following fact is central for our work and can be found in ~\cite{JiMi}. 
The proof is as well done by Mahoux in ~\cite{Mah}.
\begin{proposition}[Jimbo-Miwa] Let
$y(x)\in\cp^1\setminus\{0,1,x,\infty\}$ be the point at which
$A(x;y(x))$ and $A_{\infty}$ have a common eigenvector, corresponding
to the eigenvalue $\lambda$ of $A_{\infty}$ . If the $A_i$ satisfy 
Schlesinger's equations, then $y(x)$
satisfies the Painlev\'e equation P{\sc vi}$(\aaa,\b,\gamma,\delta)$
\begin{align*} 
      \frac{d^2y}{dx^2}= &\frac12\left(\frac1y+\frac1{y-1}+\frac1{y-x}\right)\left(\frac{dy}{dx}\right)^2-
                               \left(\frac1x+\frac1{x-1}+\frac1{y-x}\right)\left(\frac{dy}{dx}\right)+\\
                         &\left(\alpha+\beta\frac{x}{y^2}+
                      \gamma\frac{x-1}{(y-1)^2}+\delta\frac{x(x-1)}{(y-x)^2}\right)\frac{y(y-1)(y-x)}{x^2(x-1)^2},
\end{align*}
with parameters
\begin{align*}
\alpha &=\tfrac12(2\lambda-1)^2  &\gamma &=-2\text{det}(A_1)\\
\beta &=2\text{det}(A_0)         &\delta &=\tfrac12(1+4\text{det}(A_2)).
\end{align*}
\end{proposition}

\section{Relations between $\nabla$ and the
Yang-Mills connection}

In this section we make some local calculations in the general set up, mostly 
to settle the notation. Many expressions will be essential in section 4.4.

\subsection{The connection 1-form in a real twistor line} 
 Take $x\in M$ with principal orbit type, and let $p^{-1}(x)$ be the 
real twistor line
  determined by it.
  Fix $z\in p^{-1}(x)\setminus Y$ and take a parallel section of 
 $\tilde E$ with respect to the flat connection in a neighbourhood of
 $z$. Denote this section by $\tilde u$. We can take a ``locally
 equivariant'' trivialization (in the sense that it is equivariant for
 a neighbourhood of the identity) of $E$ in a neighbourhood of $x$, 
 $\{s_1,s_2\}$. Let $\Phi$ be the Yang-Mills connection 1-form
 asociated to this frame.

When one considers the frame
$\{p^{\ast}s_1\big|_{p^{-1}(x)},p^{\ast}s_2\big|_{p^{-1}(x)}\}$
one has a holomorphic trivialization of $\tilde
E\big|_{p^{-1}(x)}$. With respect to it the identity
$$T\tilde u(z)=-A(T)\tilde u(z)$$
is satisfied,
$T$ being a tangent vector to the twistor line at $z$.

Let $\{\X_i\}$ be a basis of $\mathfrak g$; then there exist 1-forms $\aaa_i$, $i=1,
2, 3$ such that $T=\aaa(\sum_i\aaa_i(T)\X_i)=\sum_i\aaa_i(T)\aaa(\X_i)$. Let
$u$ be the local section of $E$ defined by $u(p(gz))=\tilde u(gz)$ for $g$
in a suitable neighbourhood of the identity of $G$. Hence, by its definition
$u$ verifies the equation
$$\X_i\cdot u=0.$$
Using the fact that $T$ is tangent to the fibre, that $\tilde u$ is
holomorphic, and that by the definition of the holomorphic structure
on $\tilde E$ one has $(p^{\ast}{\mathcal D})^{0,1}=\delbar$ (remember
that $\mathcal D$ denotes the anti-self-dual connection), one obtains that   
\begin{align*} 
     T\tilde{u}(z) &= (p^{\ast}\nabla)_T\tilde{u}(z)=
     (p^{\ast}\nabla)_{\sum_i\aaa_i(T)\aaa(\X_i)}\tilde{u}(z)=\sum_i\aaa_i(T)(p^{\ast}\nabla)_{\aaa(\X_i)}\tilde{u}(z)\\
       &= \sum_i\aaa_i(T)p^{\ast}(\nabla_{\X_i}u)(z)=\sum_i\aaa_i(T)p^{\ast}\Phi(\aaa(\X_i))\tilde{u}(z).
\end{align*}
Then we finally arrive at the following relation between the two
connections we are working with

\begin{proposition} Let $x\in M$ be a point with principal orbit type,
  $\Phi$ be the Yang-Mills connection 1-form
with respect to a locally equivariant frame, and $\aaa_i$ the 1-forms on
$P=p^{-1}(x)$ such that 
$\aaa^{-1}=\sum_{i=1}^3\aaa_i\X_i$. Then, the connection 1-form of
$\nabla\big|_P$ is given by
\begin{equation}
A_P=-\sum_{i=1}^3\aaa_i\Phi(\X_i(x)).
\end{equation}
\end{proposition}

\subsection{The anti-self-duality equation} 
The invariance of the anti-self-dual connection imposes restrictions on
the degrees of freedom of $\Phi$. For example, because the orbit space
$M/G$ is one dimensional, the anti-self-dual equation is reduced to an
ordinary differential equation in one variable. In what follows we
will use this fact to describe our connections in a more explicit way.

Remember that it is possible to take a section $c\colon(0,1)\rightarrow
M$ intersecting perpendicularly each three dimensional orbit. Take an
open set $U$ of the class of the identity in $G/\Gamma$ such that
$G\rightarrow G/\Gamma$  is trivial as a principal bundle. We can
choose a frame $\{s_1, s_2\}$ of $E$
on $(0,1)\times U$ in such a way that the following is satisfied: 
$${\mathcal D}_{\dot{c}}s_i=0,\quad\text{and}\quad g\cdot s_i(c(t))=s_i(g\cdot c(t)).$$
 Let now $\Phi$
be the connection 1-form of $\mathcal D$ with respect to such a frame; since 
the
chosen frame is equivariant equation (3.1) is satisfied. Because of the
invariance we have that $\mathcal D$ is determined (in the complement of
the exceptional orbits) by the map
$$\phi\colon(0,1)\longrightarrow\mathfrak s\otimes\mathfrak g^{\ast}$$
given by  $\phi_t(\X)=\Phi_{c(t)}(\X(c(t)))$; $\mathfrak s$ is the Lie
algebra of the structure group of $E$.

For each  $\gamma\in\Gamma$ we define a new frame $\{s_1^{\gamma},
 s_2^{\gamma}\}$ of $E$ by the expression
 $s_i^{\gamma}(x)=\gamma\cdot s_i(\gamma^{-1}x)$. The connection
 1-form of  $\gamma\cdot{\mathcal D}$ with respect to this new frame is
 given by $\Phi^{\gamma}=\gamma^{\ast}\Phi$.

 Now, if $\la_{\gamma}\colon(0,1)\times U\rightarrow S$ is the change
 of basis matrix from $\{s_i\}$ to $\{s_i^{\gamma}\}$, using the
 invariance of $\mathcal D$ we have that
$$\gamma^{\ast}\Phi=\Phi^{\gamma}=\la_{\gamma}(x)\Phi\la_{\gamma}(x)^{-1}+\la_{\gamma}(x)d\la_{\gamma}(x)^{-1}.$$
It is easy to see that $\la_{\gamma}$ is constant. First let us see that
it does not depend on $t$:
$${\mathcal D}_{\dot{c}}s_i^{\gamma}={\mathcal D}_{\dot c}(\gamma
 s_i\gamma^{-1})=(\gamma\cdot{\mathcal D})_{\dot c}s_i={\mathcal D}_{\dot c}s_i=0,$$
and on the other hand
$${\mathcal D}_{\dot
  c}s_i^{\gamma}=\sum_j\left(\frac{\partial\la_{\gamma}}{\partial
  t}\right)_{ij}s_j;$$
puting both expressions together we have
  $\frac{\partial\la_{\gamma}}{\partial t}=0$. It is staightforward to
  see that for fixed
 $t$  the equation $\la_{\gamma}(x)=\la_{\gamma}(c(t))$ is satisfied for all
  $x\in\{t\}\times U$. Then, $\la_{\gamma}$ is constant and
$$\gamma^{\ast}\Phi=\la_{\gamma}\Phi\la_{\gamma}^{-1}.$$

Applying the above result at the point $c(t)$ we obtain
\begin{equation}
\text{Ad}_{\la_{\gamma}}(\phi_t(\X))=\phi_t(\text{Ad}_{\gamma}(\X))\qquad
\forall\gamma\in\Gamma.
\end{equation}
To find the anti-self-dual equation in terms of
$\phi_t$ let us write
$$\phi_t=\phi_1\otimes\sigma_1+\phi_2\otimes\sigma_2+\phi_3\otimes\sigma_3,$$
where $\phi_i\colon(0,1)\rightarrow\mathfrak s$ and $\{\sigma_1, \sigma_2,
\sigma_3\}$ is the dual basis to the one chosen for $\mathfrak g$. The
curvature of 
$\mathcal D$ evaluated at  $c(t)$ is then given by 
\begin{align*}
F^{\mathcal D}_t&=(d\Phi)_{c(t)}-\Phi_{c(t)}\wedge\Phi_{c(t)}\\
            &=\sum_i\dot\phi_idt\wedge\sigma_i+\sum_{i<j}\left(\sum_kC^k_{ij}
                \phi_k-[\phi_i,\phi_j]\right)\sigma_i\wedge\sigma_j,\end{align*}
where $C^k_{ij}$ are the structure constants of $G$. Taking $c$ such that
$\{\dot c, \X_1, \X_2, \X_3\}$ is a positively oriented basis for
$TM$ we will have that the anti-self-dual equation $\ast
F^{\mathcal D}_t=-F^{\mathcal D}_t$ can be written as 
\begin{equation}
K_1\dot\phi_1=-\sum_{\ell=1}^3C_{23}^{\ell}\phi_{\ell}+[\phi_2,\phi_3],\ldots
\end{equation}
$K_1=\frac{||\X_2|||\X_3||||dt||}{||\X_1||}$ being the coefficient
coming from the fact that the chosen basis is not orthonormal, plus
other two equations that are obtained by cyclic permutations of the indexes 1,2,3.

\begin{Rem} As we had seen before  $\text{tr}(A_{\infty}^2)$ is
constant in the deformation. Using this fact together with equation
(3.1) we obtain that the function
\begin{equation}
\sum_{ij}\aaa_{i,\infty}\aaa_{j,\infty}\text{tr}(\phi_i\phi_j)=\text{tr}
(A^2_{\infty})=2\la^2
\end{equation}
is an integration constant for equation (3.3); we denote by 
$\aaa_{i,\infty}(t)$ the residue at $\infty$ of the 1-form $\aaa_i$ on the
line $P_t$. The same can be done with the other residues, but because of
relations (2.5) between them only two independent constants
of movement remain (in principle).
\end{Rem}

\section{Painlev\'e's equation and
$\Su_2$-invariant instantons on $S^4$}

\subsection{The $\Su_2$ action on $S^4$}

The $\Su_2$ action on $S^4$ that we will consider is that coming from
the irreducible representation in $\R^5$, which is obtained naturally by
taking the $\So_3$ conjugacy action on trace free symmetric
$3\times 3$ matrices with real coefficients. The generic orbit is
three dimensional and there are two exceptional orbits of dimension
two, then $S^4/\Su_2\cong[0,1]$; for this description consult for
example ~\cite{Gil1}. For our purposes it is more convenient to describe it from
another point of view, starting with the complex four dimensional irreducible
representation of $\Su_2$.  We identify $\C^4$ with the space of
homogeneous polynomials of degree three, in two variables, and with
complex coefficients

$$
\C^4\cong\left\{\p(\x,\y)=z_1\x^3+\sqrt3z_4\x^2\y+\sqrt3z_3\x\y^2+z_2\y^3 
: z_i\in\C, i=1,2,3,4\right\}.
$$  
The action is as usual: $\Su_2$ acts on the variables by matrix multiplication

$$ 
\begin{pmatrix}            a       &     b\\
                      -\bar b  &  \bar a
        \end{pmatrix}\cdot (\x,\y)
   = (\x,\y)\begin{pmatrix}            a    &     b\\
                            -\bar b  &  \bar a
            \end{pmatrix}^t
   =(a\x+b\y, -\bar b\x+\bar a\y),
$$
and on polynomials by precomposition $g\cdot\p(\x,\y)=\p(g^{-1}\cdot(\x,\y))$.
The above representation is quaternionic if we identify $\C^4$ with
$\Hl^2$ by the map
$$ (z_1,z_2,z_3,z_4)\longmapsto(z_1+z_2j,z_3+z_4j),
$$
and preserves the canonical bilinear form of $\Hl^2$ (this is what
the $\sqrt3$ factors  are
necessary for). We are considering $\Hl^2$ as a left $\Hl$-module.

When we consider the action on quaternionic lines we obtain the $\Su_2$
action on $S^4$ (remember the identification $\hp^1\cong S^4$ ), and
looking at the action on complex lines we obtain the action on 
$\cp^3$ (the twistor space of $S^4$). Denote by $[h_1,h_2]_{\C}$ a
point in  $\cp^3$ and by $[h_1,h_2]_{\Hl}$ the corresponding point
under twistor projection  in  $\hp^1$, where
$(h_1,h_2)\in\Hl^2$. The description of the twistor space of $S^4$ can
be found in ~\cite{At1}.

In this example the divisor $Y$ is the quartic in $\cp^3$ of cubic
polynomials having a repeated root, i.e. the zero locus of the
discriminant. The $\operatorname{SU}_2$-action on $\cp^3$ can be extended to an
$\operatorname{SL}_2(\C)$-action. Writing a polynomial in the form
$\p(\x,\y)=\y^3\q(\frac{\x}{\y})$, $\operatorname{SL}_2(\C)$ acts
in $\q$ by M\"obius transformations in its roots. There are three
different orbits, one is the dense one of dimension three $\cp^3\setminus
Y\cong\Sl/S_3$ ($S_3$ is the permutation group of three elements),  and 
the others are
of dimensions two and one. From this point of view the set $Y$ is the
union of the lower dimensional orbits.

A line in $\cp^3$ is determined by two linearly independent polynomials $\q_1,
\q_2$. The points of intersection with $Y$ correspond to those polynomials
of the form $\q_1+\zeta\q_2$ which have zero discriminant, namely,
 those for which $z$ is such that the system
\begin{align*} 
           \q_1+\zeta\q_2&=0\\
         \q'_1+\zeta\q'_2&=0
\end{align*}
has a solution. Therefore $\zeta=-\frac{\q_1(\aaa)}{\q_2(\aaa)}$, for
a root $\aaa$ of  $\q'_1\q_2-\q_1\q'_2$. For a generic line this
polynomial has degree four, with four distinct roots, each of
which gives a point of intersection of the line with $Y$. So, we have
that the divisor intersects a generic line in four different points.

The two bidimensional orbits in $S^4$ are those of the points
$x_0=[1,0]_{\Hl}$ and $y_0=[0,1]_{\Hl}$, whose stabilizer is the subgroup 
$$ \widetilde{\OO}_2=
\left\{\begin{pmatrix}  a  &    0\\
                    0  &  \bar a \end{pmatrix} : |a|^2=1\right\}\cup
                \left\{\begin{pmatrix}     0    &    b\\
                               - \bar b  &    0 \end{pmatrix} : |b|^2=1\right\}\subset\Su_2
$$
(we denote the elements of this subgroup simply by $a$ in the diagonal
case, and by $b$ in the other case), what can be easily verified. In fact
we have
\begin{align*} 
a\cdot\x^3 &=\phantom{-}\bar a^3\x^3 &\qquad a\cdot\x\y^2=\phantom{-}a\x\y^2\\
          b\cdot\x^3 &=-b^3\y^3        &\qquad b\cdot\x\y^2=-\bar
          b\x^2\y
\end{align*}

In view of this the two exceptional orbits are isomorphic to 
$\Su_2/\tilde{\OO}_2\cong\So_3/\OO_2\cong\rp^2$, and we will denote
them by $\rp^2_{\pm}$.  
The real twistor lines over the points $x_0$ and $y_0$ are given by 
$$ p^{-1}(x_0)=\left\{[h,0]_{\C} : h\in\Hl\right\}\qquad
   p^{-1}(y_0)=\left\{[0,h]_{\C} : h\in\Hl\right\},$$
and they have as stabilizer, under the $\Sl$ action, the subgroup
$$ \left\{\begin{pmatrix}  a  &    0\\
                    0  &   a^{-1} \end{pmatrix}  : a\in\C^{\ast}\right\}\cup
                \left\{\begin{pmatrix}     0    &    b\\
                               - b^{-1}  &    0 \end{pmatrix} : b\in\C^{\ast}\right\};
$$ 
what can be proved with a simple calculation. Moreover, in each of
these lines there are only two points with non discrete stabilizer,
which tells us that they intersect $Y$ in exactly two points. In
particular none of them is transversal to $Y$. The intersection can be
written explicitly in the following way:
$$p^{-1}(x_0)\cap Y=\{[1,0]_{\C};[j,0]_{\C}\}\qquad p^{-1}(y_0)\cap Y=\{[0,1]_{\C};[0,j]_{\C}\}.
$$

Observe that the stabilizer of the point $[1,0]_{\C}$ is the subgroup
of lower triangular matrices and the stabilizer of $[j,0]_{\C}$ is the
subgroup of upper triangular matrices, while that of the points
$[0,1]_{\C}$ and $[0,j]_{\C}$ is the subgroup of diagonal
matrices. Therefore, the first two are in the singular part of $Y$ and
the other two in the smooth one (i.e. in the two
dimensional orbit).

The subgroup $\Gamma=\{\pm 1,\pm i,\pm j,\pm k\}\subset\Su_2$ has
as fixed point set the maximal circle
$\Sigma:=\{[s,t]_{\Hl} : s,t\in\R\}$ (which intersects perpendicularly
all the orbits),
and we can see that 

\begin{align*} \sigma(x^+)\cap\Sigma &=\{[1,0]_{\Hl},[1,\sqrt3]_{\Hl},[-1,\sqrt3]_{\Hl}\} \\
          \sigma(x^-)\cap\Sigma
          &=\{[\sqrt3,1]_{\Hl},[0,1]_{\Hl},[-\sqrt3,1,]_{\Hl}\},
\end{align*}
where $x^+=[1,0]_{\Hl}=x_0$ and $x^-=[\sqrt3,1]_{\Hl}$.

We parametrize the segment on $\Sigma$ from $x^{+}$ to $x^{-}$ by
$c\colon [0,1]\rightarrow S^4$, given by $c(t)=[\sqrt3,t]_{\Hl}$ (this
parametrization is not geodesic but it simplifies some calculations).
This segment intersects each orbit in exactly one point and its ends
are on the exceptional orbits. For each $t$ we have a real twistor line
that we will denote by $P_t$:
$$P_t=p^{-1}(c(t))=\{[z(\x^3+t\x\y^2)+w(t\x^2\y+\y^3)]_{\C}:z,w\in\C\}\subset\C
P^3,$$
or equivalently  $P_t=\{(\sqrt3z:\sqrt3w:tz:tw)\in\C P^3 : z,w\in\C\}$.

In order to obtain the intersection $Y\cap P_t$ we do as above. The
polynomials $\q_1, \q_2$ corresponding to the line $P_t$ are
$\q_1(\x)=\x^3+t\x$, $\q_2(\x)=t\x^2+1$. Then, the intersection points
are those for which
$\zeta=-\frac{\q_1(\aaa)}{\q_2(\aaa)}$, where $\aaa$ is a root of
$$\q'_1\q_2-\q_1\q'_2=t\x^4+(3-t^2)\x^2+t.$$
The last polynomial has four roots if
$t^4-10t^2+9=(t^2-1)(t^2-9)\neq0$ and $t\neq0$, then $|P_t\cap
Y|=4$ when $t\in(0,1)$ and $|P_0\cap Y|=|P_1\cap Y|=2$. 

\subsection{Description of the invariant instantons over $S^4$}

If we want to apply the description described in section 2, and so relate solutions
to the Painlev\'e sixth equation with anti-self-dual solutions to the
Yang-Mills equation, we need to have rank two vector bundles over
$S^4$ which are $\Su_2$-equivariant with respect to the action
descibed in 4.1. On the other hand we require that they admit an ASD
invariant connection. For the construction of such bundles we refer to
~\cite{Gil1} and ~\cite{Gil2}. In the present section we will give a
description of the results. Remember that the orientation of $S^4$
considered here is the standard one of $\hp^1$, with
respect to which the Hopf bundle is self-dual ($c_2=-1$).

Suppose that $E\rightarrow S^4$ is a rank two Hermitian vector bundle
on which $\Su_2$ acts covering the action on $S^4$ described in 4.1.
Our action on $S^4$ has two exceptional orbits, of dimension two,
corresponding to the points $x^+$ and $x^-$. The fibres of $E$ over
these points have to be invariant under their respective stabilizers,
then they are two dimensional
representations of $\widetilde{\OO}_2\subset\Su_2$. By means of a simple exercise it is
possible to find all the representations of this kind, which are
classified by the integers congruent with 1 modulo 4 and zero (the
trivial representation). Restricted to the diagonal subgroup
$S^1\subset\widetilde{\OO}_2$ each non trivial representation splits
as $\C_{(r)}\oplus\C_{(-r)}$, for a positive odd integer or zero $r$,
$\C_{(r)}$  being the one dimensional representation of $S^1$ with weight
$r$. According to the previous considerations, each lifting of the
action on $S^4$ to a rank two vector bundle defines a pair of integers
$(n_{+},n_{-})$, where $n_{\pm}\equiv 1\text{mod }4,\,0$.

Conversely, it is proved that such a pair of integers uniquely determines
(modulo $\Su_2$-isomorphisms) a vector bundle $E\rightarrow S^4$
together with an $\Su_2$ action (by morphisms) covering the given
action on $S^4$. The vector bundles are obtained from the trivial one
over the open sets $U^{\pm}:=S^4\setminus\rp^2_{\pm}$ by means of a
``clutching construction''. We summarize the previous dicussion in the
following proposition. For a detailed proof consult ~\cite{Gil1}.

\begin{proposition}[Bor]  For each pair of integers $(n_+,
             n_-)$ as above, there exists a complex vector bundle 
               $E_{(n_+, n_-)}\rightarrow S^4$ with $\operatorname{SU}_2$-action,
             lifted from $S^4$, such that $n_+$ and $n_-$ are the weights 
              of the representations of the singular  stabilizers on
             the corresponding fibres $E_{x^+}$ and $E_{x^-}$.
              Moreover, the Chern number of the vector bundle
             corresponding to $(n_+, n_-)$ is given by 
              $$c_2=\frac{n_-^2-n_+^2}{8}.$$
\end{proposition}

It is possible to prove that those vector bundles for which $n_+>1$
and $n_->1$ do not admit any self-dual or anti-self-dual
$\Su_2$-invariant connection (~\cite{Gil1}). But using an equivariant
version of the ADHM construction one can obtain the following
statement ~\cite{Gil3}:

\begin{proposition} For each  $n=1\text{{\rm mod} }4$ there exists a unique
$\Su_2$-invariant anti-self-dual connection on
              $E_n:=E_{(1, n)}$ {\rm (}i.e., for $n_+=1${\rm )}.
\end{proposition}

And these vector bundles are just what we needed in order to apply the
general construction mentioned in section 2, since as we have seen, 
the lines in
the family $P_t$ intersect the divisor in four different points for
$t\in(0,1)$, and then the hypotheses of proposition 2.5 are satified.

\subsection{The parameters of P{\sc vi}}

Using the $\Sl$ action it is possible to calculate  the parameters of
P{\sc vi} corresponding to the invariant instantons described in the
previous section in a geometric way, though they can be found also
by means of the expression (3.4) as we will see later.
Let $z\in Y^{\circ}$, $U$ be a trivializing neighbourhood of $E_n$ and
$(z_1,z_2,z_3)$ be coordinates on $\C P^3$ such that 
$Y$ is defined by  $z_1=0$. Consider a rational curve $\xi\colon\C
P^1\rightarrow \C P^3$  transversal to $Y$ at $z$. The 1-form
$\xi^{\ast}A$, defined on  $\xi^{-1}(U\setminus Y)$, has logarithmic
pole at $\xi^{-1}(z)$ with residue $B(z)$, i.e. equal to the residue
of $A$ evaluated at the point $z$. Thus, to find the residue of $A$ at
a given point it is sufficient to know $A$ restricted to a rational
line, transversal to $Y$ at that point. To find the parameters of
Painlev\'e's equation we have to calculate the residues, or
equivalently the trace of its squares. Observe that if $z\in
Y^{\circ}$ and $g\in \Sl$, then the residue of the connection at the
point $gz$ is conjugate to the residue at $z$ since as the connection is
$\Sl$-invariant  we can take a trivialization of $E_n$ around $gz$
for the connection 1-form to be $g^{\ast}A$; thus, the trace of
the squares of the residues does not depend on the point of
$Y^{\circ}$ at which we calculate it. From these two observations we conclude
that it is sufficient to take a point in $Y^{\circ}$, a twistor line
transversal to $Y^{\circ}$ at that point, and calculate the residue of
the connection restricted to the line at such point. For that we will
fix $z_0=[0,1]_{\C}\in Y^{\circ}$; the stabilizer of 
$z_0$, as was mentioned before, is $\left\{\left(\smallmatrix a & 0\\
0&\frac 1a\endsmallmatrix\right) : a\in\C^{\ast}\right\}\cong\C^{\ast}$.
The group $\C^{\ast}$ acts on $\C^4$ in the following way
\begin{align*} 
            a\cdot\x^3 &=a^{-3}\x   &\qquad a\cdot\x^2\y =a^{-1}\\
         a\cdot\x\y^2  &=a\x\y^2    &\qquad \phantom{\x}a\cdot\y^3=a^3\y^3. 
\end{align*}
When we projectivize we have three lines, invariant under the
$\C^{\ast}$ action, that pass through $z_0$; these lines are
$\zeta \mapsto [\x\y^2+\zeta \x^2\y]$, $\zeta
\mapsto[\sqrt3\x\y^2+\zeta \y^3]$ and $\zeta \mapsto[\sqrt3\x\y^2+\zeta
\x^3]$ (with $\zeta \in\C P^1\setminus\{\infty\}\cong\C$). The first
two have weights -2 and 2 respectively, so they correspond to tangent
lines to $Y^{\circ}$, and the third one,
which is then transversal, has weight -4.

For the calculation we need to know the $\C^{\ast}$-vector bundles
over $\C$, i.e. the holomorphic vector bundles over $\C$ that come
with a $\C^{\ast}$ action. The classification of these objects is
simple and is given by the following lemma that can be found in ~\cite{At2}.

\begin{lemma}[Atiyah] Let $E$ be a $\C^{\ast}$-holomorphic vector
                      bundle over $\C$. Then $E\cong\C\times E_0$ with
              the diagonal action on $\C\times E_0$:
              $\lambda\cdot(w,v)=(\lambda^k w,\lambda\cdot v)$.
\end{lemma}

We are now in position to establish the following theorem that says
which are the parameters of P{\sc vi} related to the family of vector
bundles described in 4.2.

\begin{theorem} Let $\{E_n: n\,\text{odd}\}$ be the family of
            $SU_2$-invariant instantons over $S^4$ 
            described in 4.2. Then, the Painlev\'e equation
            corresponding to the isomonodromic deformation obtained
            from the Penrose-Ward transform of $E_n$, with the
            connection defined by the action, has the following parameters:
            $$\aaa^{\pm}=\tfrac 18(n\pm 2)^2,\quad\b=-\tfrac
            18n^2,\quad\gamma=\tfrac 18n^2,\quad\delta=-\tfrac
            18(n^2-1).$$   
\end{theorem} 

\begin{proof} Let $\xi\colon\C P^1\rightarrow Z$ be the rational curve given by
            the expression $\xi(t)=[\zeta,1]_{\C}$; we know that it is
 transversal to $Y$  at the point $z_0=[0,1]_{\C}$. Restricting the
 vector bundle to the line we obtain a $\C^{\ast}$-holomorphic vector
 bundle over $\cp^1$, and in view of the previous lemma we have that
$$
(\xi\big|_{\C P^1\setminus\{\infty\}})^{\ast}E\cong\C\times E_{n,z_0}
$$
with the $\C^{\ast}$ action given by
 $\lambda\cdot(\zeta,v)=(\lambda^{-4}\zeta,\lambda\cdot v)$. Moreover,
 by the construction,  $E_{n,z_0}$ is a representation of $\C^{\ast}$
 with weights $-n, n$; therefore, taking a base of eigenvectors, we
 have a $\C^{\ast}$-isomorfism between $\C\times E_{n,z_0}$ and
 $\C\times\C^2$ with the action $\lambda\cdot(\zeta,\lambda_1,\lambda_2)=(\lambda^{-4}\zeta,\lambda^{-n}
\lambda_1,\lambda^n\lambda_2)$. Then the holomorphic field generated
 by the action can be expressed as follows
$$
X=-4\zeta\frac{\partial}{\partial \zeta}-n\lambda_1\frac{\partial}{\partial \lambda_1}+n\lambda_2
    \frac{\partial}{\partial \lambda_2}.
$$
Recall that this field has to be horizontal because of the definition of
the connection. Now, the 1-forms that anihilate the horizontal
subspace are
$$
\theta_{\aaa}=d\lambda_{\aaa}+\sum_{\b}\lambda_{\b}A_{\aaa\b},\qquad\aaa=1,2;
$$
(see proof of theorem 2.1 in ~\cite{AHS}). 
 Thus we have that 
$$\quad0=\theta_{\aaa}(X)=(-1)^{\aaa}\lambda_{\aaa}-4\zeta\sum\limits_{\b}\lambda_{\b}A_{\aaa\b}(\zeta),$$ 
from what we obtain 
$A_{12}=A_{21}=0$ and $A_{11}=-A_{22}=\dfrac{n}{4\zeta}$, and we
conclude that the invariant we wanted to find is
$$
\operatorname{tr}(\operatornamewithlimits{Res}_{z_0}A)^2=\tfrac 18n^2.
$$ 
This, together with proposition 2.5 finishes the proof.
\end{proof}

In the next section we will see another proof of the previous theorem, based 
on more explicit calculations; see remark 4.1.

\subsection{Some explicit calculations}

In this section we will calculate explicitly the connection form of
the flat connection, restricted to the family of lines we are
considering, in terms of the Yang-Mills connection. This will give us
a better understanding of what is happening and how to obtain the
solution to  Painlev\'e's equation from the solution to the
anti-self-dual equation.

Remember that for any real twistor line $P$ we can write the flat
connection in the following way (see equation (3.1))
$$A(T)=-\sum_{i=1}^{3}\aaa_i(T)p^{\ast}\Phi(\aaa(\X_i)),$$
where $\{\X_1,\X_2,\X_3\}$ is a base for $\su_2$, $T$ is a tangent
vector to $P$, $\Phi$ is the Yang-Mills connection 1-form on $S^4$,
and the 1-forms $\aaa_i$ are defined by
$$\aaa^{-1}(T)=\sum_{i=1}^3\aaa_i(T)\X_i.$$
The map $\phi\colon(0,1)\rightarrow\su_2\otimes\su_2$ defined in
section 3.2,  which determines the flat connection, has to be of the form
$$\phi_t=a_1(t)\X_1\otimes\sigma_1+a_2(t)\X_2\otimes\sigma_2+a_3(t)\X_3\otimes\sigma_3$$
because of its invariance under the $\Gamma$ action (see equation (3.2)),
with $a_i\colon(0,1)\rightarrow\R$.

The (anti-)self-dual equation can be written in terms of the $a_i$'s
 in the form
\begin{equation}
(-)\frac12 K_1\dot{a}_1=a_2a_3-a_1,
\end{equation}
plus other two equations obtained by cyclic permutation of the indexes
1,2,3. The coefficients $K_i$, coming from the fact that the basis
formed by the generating vectors of the action and $\dot c$ is not
orthonormal, are given by the expressions
$$K_1(t)=\frac{(t^2-1)(t^2-9)}{4t},\quad K_2(t)=4t\frac{(t-3)(t+1)}{(t+3)(t-1)}
,\quad
K_3(t)=4t\frac{(t+3)(t-1)}{(t-3)(t+1)}.$$
In fact, the functions $a_i$ can be extended to all of $\Sigma$ (as in
~\cite{Gil2}) and they satisfy the boundary conditions: $a_1(0)=1,
a_2(1)=-n, a_1(1)=a_3(1)=0$ ~\cite{Gil2}.

\begin{example} For the trivial bundle $E_1$ the anti-self-dual connection is
determined by the functions $a_1(t)=a_2(t)=a_3(t)=1$ (of course, it is
also self-dual).\end{example}

\begin{example} The Hopf bundle, which is self-dual, has the $a_i$'s given by
the expressions
$$a_1(t)=\frac{t^2-9}{t^2+3}, \qquad
a_2(t)=-2t\frac{(t+3)}{t^2+3},\qquad a_3(t)=-2t\frac{(t-3)}{t^2+3}.$$
Applying the antipodal map one sees that the anti-self-dual form on
$E_{-3}$ is thus given by 
$$a_1(t)=3\frac{(1-t^2)}{t^2+3}, \qquad
a_2(t)=-6\frac{(t+1)}{t^2+3},\qquad a_3(t)=-6\frac{(t-1)}{t^2+3}.$$
\end{example}
Our next step is to find $\aaa^{-1}$. In order to do so we take homogeneous
 coordinates on $\cp^3$
 $(\la,\mu,\zeta)=(\frac{z_1}{z_2},\frac{z_3}{z_2},\frac{z_4}{z_2})$,
in the open set $\{z_2\neq 0\}$,
associated with which we have a tangent basis on $\cp^3$
\newcommand{\deriva}[1]{\frac{\partial}{\partial#1}}
that we denote by
$\left\{\deriva\la,\deriva\mu,\deriva\zeta\right\}$. As a basis for
$\su_2$ we take
$$ \X_1=\begin{pmatrix} i & 0\\
                0 & -i\end{pmatrix}\qquad \X_2=\begin{pmatrix} 0 & 1\\
                                                    -1 &0 \end{pmatrix}
\qquad \X_3=\begin{pmatrix} 0 & i\\
                    i & 0\end{pmatrix}. 
$$
As an intermediate step it is convenient to take a complex basis for 
$\mathfrak{sl}_2(\C)$ defined by
$\left\{-i\X_1, \frac12(\X_2+i\X_3),-\frac12(\X_2-i\X_3)\right\}$,
with respect to which the matrix associated to $\aaa$ at the point
$(\la,\mu,\zeta)$ is 

\begin{equation}
\begin{pmatrix} -6\la   &   -\sqrt3\la\mu   & \sqrt3\zeta \\
           -2\mu   &     \sqrt3        &  2\zeta-\sqrt3\mu^2  \\
           -4\zeta &      2\mu         &  \sqrt3\la-\sqrt3\mu\zeta
	   \end{pmatrix}.
\end{equation}

In homogeneous coordinates our family of lines is parametrized by 
$(\la,\frac{1}{\sqrt3}t\la,\frac{1}{\sqrt3}t)$, and so the vector 
 $T=\deriva\la+\frac{t}{\sqrt3}\deriva\mu$ is tangent to $P_t$, and
 consequently, using (4.2), we can easily obtain that
\begin{equation}
 \begin{split}
\aaa^{-1}(T)=-\frac{1}{3\Delta(t,\la)}\bigl( &i(t^2-1)(t^2-9)\frac{\la}{2} 
\X_1-t(t+1)(t-3)(\la^2+1)\X_2\\
                                          &+it(t-1)(t+3)(1-\la^2)\X_3\bigr),
 \end{split}
\end{equation}
where
$\Delta(t,\la)=\frac13(8t^3\la^4-2(t^4+18t^2-27)\la^2+8t^3)=-\operatorname{det}(\aaa)(\la,\la
t/\sqrt3,t/\sqrt3)$.
From the above expression we can find without difficulty the
intersection points of the divisor with each line, which are the poles
of the flat connection. They are given by the equation
$\Delta(t,\la)=0$ for each fixed $t$; from which it follows that
$$P_t\cap Y=\{\pm\sqrt{\mu_{\pm}}\},\quad\text{where } 
            \mu_{\pm}=\frac{t^4+18t^2-27\pm\sqrt{(t^2-1)(t^2-9)^3}}{8t^3}. 
$$
We put $z_4=-z_1=\sqrt{\mu_-}, z_3=-z_2=\sqrt{\mu_+}$. Note that
$\mu_{\pm}\in\R^{-}$ and $\mu_+\mu_-=1$. We make the M\"obius transformation 
such that  $z_1\rightarrow 0$, $z_2\rightarrow 1$, and $z_4\rightarrow \infty$.
 Let $\mu=\mu_+-\mu_-$, then
the residues of the forms $\aaa_i$ are given by
\begin{equation}
      \begin{split}
         \aaa_{1,0}(t)&=\frac{i(t^2-1)(t^2-9)}{16t^3\mu}=\aaa_{1,\infty}(t)
=\overline{\aaa}_{1,x}(t)
         =\overline{\aaa}_{1,1}(t)\\
         \aaa_{2,0}(t)&=\frac{(\mu_{-}+1)t(t+1)(t-3)}{8t^2\mu
         z_1}=-\aaa_{2,\infty}(t)=\aaa_{2,1}(t)=-\aaa_{2,x}(t)\\
         \aaa_{3,0}(t)&=i\frac{(\mu_{+}-1)t(t-1)(t+3)}{8t^2\mu
         z_1}=-\aaa_{3,\infty}(t)=-\aaa_{3,1}(t)=\aaa_{3,x}(t).
      \end{split}
\end{equation}

The cross ratio $x$ of the four singular points, which is the variable
in the Painlev\'e equation, is related with $t$ by the expression
$$x=\frac{(t+1)(t-3)^3}{(t-1)(t+3)^3}.
$$
The connection 1-form of $\nabla$ will be given in each line by 
\begin{equation}
 \begin{split}
A(t,\la) &=-\sum_{i=1}^3a_i(t)\aaa_i(t,\la)\X_i\\
         &=\frac{1}{3\Delta(t,\la)}\left(i(t^2-1)(t^2-9)\frac{\la}{2}
	 a_1(t)\X_1-t(t+1)(t-3)(\la^2+1)a_2(t)\X_2\right.\\
         &\phantom{\frac{1}{3\Delta(t,\la)}}\quad\left.+it(t-1)(t+3)(1-\la^2)a_3(t)\X_3\right),
     \end{split}
\end{equation}
and the residues will then be
$$A_0(t)=-\sum_{i=1}^3a_i(t)\aaa_{i,0}(t)\X_i=-A_x^{\ast}(t)$$
$$A_{\infty}(t)=-\sum_{i=1}^3a_i(t)\aaa_{i,\infty}(t)\X_i=-A_1^{\ast}(t);$$
the identities between the residues can  be verified from
(4.6). Remember that the solution to the sixth Painlev\'e equation is
given by the function $y(x)$, where $y(x)$ is the only point in 
$\cp^1\setminus\{0,1,x,\infty\}$ such that $A(t,y(x))$ and
$A_{\infty}(t)$ have a common eigenvector corresponding to one of the
eigenvalues. Then, with the above expressions it is possible to find
$y(x)$ in terms of the instanton. 

\begin{Rem} {\it Another proof of theorem 4.4.} We have from equation 3.4
$$\operatorname{tr}(A^2_{\infty})=\sum^3_{i=1}a_i(t)^2\aaa_{i,\infty}(t)^2
\operatorname{tr}(\X_i^2)=-2\sum^3_{i=1}a_i(t)^2\aaa_{i,\infty}(t)^2.$$
From expressions (4.4) it is easy to see that $\lim_{t\to 1}\aaa_{1,\infty}(t)
=\lim_{t\to 1}\aaa_{3,\infty}(t)=0$ and $\lim_{t\to 1}\aaa_{2,\infty}(t)=
\frac{i}{4}$. Therefore, using the boundary conditions on the $a_i$'s, we have
$$\operatorname{tr}(A^2_{\infty})=-2\lim_{t\to 1}\sum^3_{i=1}a_i(t)^2
\aaa_{i,\infty}(t)^2=\frac{n^2}{8},$$
which is the same value obtained in the proof of theorem 4.4.
\end{Rem}

\bibliographystyle{plain}
\bibliography{ref_article}
\end{document}